\documentclass[a4paper,10pt]{amsart}
\usepackage[foot]{amsaddr}
\usepackage{amsmath,amsfonts,bm}
\usepackage{amsthm}
\usepackage{graphicx}
\usepackage{mathrsfs}
\usepackage{paralist}

\newcommand{\R}{\mathbb{R}}
\newcommand{\pdef}{\mathrel{\mathop:}=}

\theoremstyle{plain}		\newtheorem{prop}{Proposition}
\theoremstyle{definition}	\newtheorem{dtion}{Definition}

\DeclareMathOperator \sgn {sgn}

\usepackage[T1]{fontenc}

\title[Analysis of Compactly Constructed Time Machine Spacetimes]{Geometric Analysis of Particular Compactly Constructed Time Machine Spacetimes}
\author{J.~Dietz}
\author{A.~Dirmeier$^{*}$}
\author{M.~Scherfner$^{**}$}
\address{Dept. of Mathematics, TU Berlin, Str. d. 17. Juni 136, 10623 Berlin, Germany}
\email{$^*$dirmeier@math.tu-berlin.de}
\email{$^{**}$scherfner@math.tu-berlin.de}
\begin{document}


\begin{abstract}
We formulate the concept of time machine structure for spacetimes exhibiting a compactely constructed region with closed timelike curves. After reviewing essential properties of the pseudo Schwarzschild spacetime introduced by A. Ori, we present an analysis of its geodesics analogous to the one conducted in the case of the Schwarzschild spacetime. We conclude that the pseudo Schwarzschild spacetime is geodesically incomplete and not extendible to a complete spacetime. We then introduce a rotating generalization of the pseudo Schwarzschild metric, which we call the the pseudo Kerr spacetime. We establish its time machine structure and analyze its global properties. 
\end{abstract}

\maketitle

\section*{Introduction}
While it is easy to construct spacetimes containing closed timelike or closed causal curves (CTCs or CCCs), it is more intractable to find examples of spacetimes that are at the same time realistic, at least to some degree. Nor is there a common agreement on which properties should be included in a list classifying how close to physical reality a spacetime with CCCs is. In \cite{Ori07}, Ori presents such a list, setting constraints on what should be classified as a realistic time machine model.\\
One essential notion in General Relativity is the Cauchy development or domain of dependence of a subset $\mathcal A$ of a spacetime $M$. It is connected with the problem of formulating the Einstein equations as a well posed initial value problem and the concept of global hyperbolicity.\\
From a realistic point of view it is desirable to obtain CCCs as the unevitable consequence of a Cauchy development. However, the interior of the Cauchy development of an achronal subset $\mathcal A$ is globally hyperbolic and therefore excludes any CCCs. The most one can hope for is that there are points in the boundary of the Cauchy development which lie on CCCs. This is the central idea of time machine structure (TM-structure) for a spacetime.
\section*{TM-Structure}
We will consider $M$ to be a four-dimensional manifold and a metric tensor $g$ on $M$ with signature $(-+++)$, such that the Lorentzian manifold $(M,g)$ together with a fixed time-orientation is a spacetime. Let $\mathcal V(M)$ denote the set of all points $p \in M$ that lie on some closed causal curve in $M$. We shall refer to $\mathcal V(M)$ as the causality violating region or the time machine. If $p \in \mathcal V(M)$, we say that causality is violated at $p$. The causality condition is said to hold on a subset $\mathcal U \subset M$ if $\mathcal V(M) \cap \mathcal U = \emptyset$.

\begin{dtion} \label{def-TM-structure}
 We say that $(M,g)$ has TM-structure if the following three properties hold for $M$:
\begin{enumerate} [(TM1)]
 \item There exists an open subset $\mathcal U$ of $M$ on which the causality condition holds,
 \item there is a spacelike hypersurface $\mathscr S$ contained in $\mathcal U$, such that
 \item for an achronal, compact subset $\mathscr C \subset \mathscr S$ the future domain of dependence of $\mathscr C$ contains points in its boundary at which causality is violated, i.e. $\overline{D^+(\mathscr C)} \cap \mathcal V(M) \neq \emptyset$. In this case the causality violating region $\mathcal V(M)$ is said to be \textit{compactly constructed}.
\end{enumerate}
\end{dtion}

This seems to be a reasonable condition for any spacetime that may be viewed as a time machine model. Although there definitely are many more additional conditions one could impose, we will focus on TM-structure here. It should be noted that models with TM-structure can satisfy the weak, strong and dominant energy condition (see \cite{HawEllis}) for the matter content of the spacetime if properly constructed. Indeed the metrics analyzed here are all vacuum spacetimes and also the TM-models constructed in \cite{Ori94} and \cite{Ori05} obey all energy conditions. As discussed in \cite{Ori93} and \cite{Ori07} this does not contradict the findings by Hawking in \cite{Hawking1992}, which introduced the concept of a \textit{compactly generated} causality violating region, which differs from condition (TM3). Hawking proved that a compactly generated time machine necessarily violates the weak energy condition.

\section*{Pseudo Schwarzschild Spacetime}
The pseudo Schwarzschild spacetime was introduced in \cite{Ori07}, where it is presented as part of a viable time machine model. As such, Ori uses the pseudo Schwarzschild spacetime as the internal core of a time machine model and matches it to an intermediate dust region which in turn is enclosed by an outer asymptotically flat vacuum region. In this section we will be occupied with the vacuum core, represented by the pseudo Schwarzschild spacetime and, after describing its geometry and some of its structure, we focus on the behaviour of geodesics.\\

The geometry of the pseudo Schwarzschild spacetime can be written in warped product form with $Z\pdef S^1 \times \R^+$ as base and the hyperbolic plane $H^2$ as fiber. On the cylinder $Z$ we use coordinates $(\nu, r)$, where $\nu$ is taken to be a circular coordinate (angle) on $S^1$ and $r$ represents the identity on $\R^+$. The metric on $Z$ then reads
\begin{equation*}
 g_1 \pdef \left( 1- \frac{2m}{r} \right) d\nu^2 + 2 d\nu dr
\end{equation*}
with an arbitrary positive constant $m$.\\
On the hyperbolic plane $H^2$ we introduce the so-called pseudo spherical polar coordinates $(\theta, \phi)$
in which the metric of the hyperbolic plane becomes
\begin{equation*}
 g_2 \pdef  d\theta^2 + \sinh^2(\theta) d\phi^2 .
\end{equation*}
Strictly speaking, this is the metric of the hyperbolic plane with Gaussian curvature $-1$. Note that here $\theta \in (0,\infty)$ and $\phi \in S^1$ is another circular coordinate. If we think of $H^2$ to be represented topologically by $\R^2$, pseudo spherical polar coordinates are described by the mapping
\begin{equation*} \label{equ-Pseudo-sph-polar-coord}
 (0,\infty) \times S^1 \to \R^2, \quad (\theta , \phi) \mapsto (\sinh \theta \cos \phi, \sinh \theta \sin \phi)
\end{equation*}
which makes clear that $(\theta,\phi)$ cover all of $H^2$ except the origin $(0,0)$.
Expressed as a warped product with warping function $r$, the pseudo Schwarzschild spacetime is
\begin{equation*}
 M \pdef Z \times_r H^2
\end{equation*}
and joining the coordinate systems for the cylinder $Z$ and the hyperbolic plane $H^2$ allows us to 
write down the pseudo Schwarzschild metric:
\begin{equation*} \label{equ-metric}
 g =  \left( 1- \frac{2m}{r} \right) d\nu^2 + 2 d\nu dr + r^2 ( d\theta^2 + \sinh^2(\theta) d\phi^2 )
\end{equation*}
This expression bears several similarities to the usual Schwarzschild metric in (outgoing) Eddington-Finkelstein coordinates. In the Schwarzschild metric, the fiber is not the hyperbolic plane but the sphere $S^2$ for which spherical coordinates are generally used. Moreover, instead of the cylinder $Z$, the base of the Schwarzschild metric is the half plane $\R \times \R^+$ with the negative of our metric for $Z$.
By analogy to the Schwarzschild metric we define the horizon $H$ to be the set of points satisfying $r=2m$ and divide the spacetime into the regions $R_1$ and $R_2$, where $r>2m$ and $r<2m$, respectively. Away from the horizon
$H$ we may undo the Eddington-Finkelstein transformation by considering
\begin{equation}\label{newcoord}
 M\setminus H \to M\setminus H, \qquad (\nu,r,\theta,\phi) \mapsto (\bar \nu, \bar r, \theta, \phi), \qquad \text{where} \quad
 \begin{cases} 
  \; \bar \nu 		\pdef -\nu - r^* \\
  \; \bar r 	\pdef r
 \end{cases}
\end{equation}
and $r^* \pdef r + 2m \ln \left| \frac{r}{2m} -1 \right|$. Note that $\bar \nu$ is still a circular coordinate.
By virtue of this transformation, the metric is converted into
\begin{equation} \label{newmetric}
 g = \left(1-\frac{2m}{\bar r} \right) d\bar\nu^2 -\left(1-\frac{2m}{\bar r} \right) ^{-1} d\bar r^2 + \bar r^2(d\theta^2 + \sinh^2(\theta) d\phi^2).
\end{equation}
Here again, we discern some structural resemblance to the Schwarzschild metric in usual coordinates.
Yet there are important differences. Unlike the radial coordinate in the Schwarschild metric, $\bar r$ is here timelike throughout $R_1$ and spacelike in $R_2$. We will take the time orientation of the pseudo Schwarzschild spacetime to be such, that $-\partial_{\bar r}$ is future pointing where it is timelike. This implies that any future pointing causal curve $\alpha$ in $R_1$ approaches the horizon $H$, since 
\begin{equation*}\label{prop5pr1}
 g(\alpha ' , -\partial_{\bar r} ) = \left( 1- \frac{2m}{\bar r} \right) ^{-1} \dot{\bar r} < 0.
\end{equation*}
Hence, $\dot {\bar r} = \dot r <0$ in $R_1$ (here, the dot refers to differentiation w.r.t.~the curve parameter).

The presentation so far shows that the coordinate singularity $\bar r =2m$ in the last representation of the metric is easy to overcome. But at $r=0$ we are confronted with a real geometric singularity as a computation of the Kretschmann scalar shows. One obtains
\begin{equation*}
 R_{abcd}R^{abcd}=\frac{48m^2}{r^6}
\end{equation*}
which obviously diverges for $r \to 0$.

The symmetries of the pseudo Schwarzschild spacetime derive from its warped product structure. Apart from the symmetries resulting from the fact that the hyperbolic plane $H^2$ is homogeneous and isotropic, which gives rise to three Killing fields, the fourth and last Killing field is given by $\partial_\nu$ since the metric coefficients do not depend on $\nu$.

As noted in \cite{Ori07}, every point of the region $R_2$ is situated on some CTC, namely on one of the curves of constant $r, \theta, \phi$, whereas the region $R_1$ contains no points on CTCs as it is foliated by the spacelike hypersurfaces $r=const$. Furthermore, Ori establishes the TM-structure of the pseudo Schwarzschild spacetime.
If we choose $r_0 \pdef 3m$, he shows that the future domain of dependence of the compact subset $\mathscr C \pdef \{ p \in \mathscr S \colon \theta(p) \leq \theta_0 \}$ with $\theta_0 > \ln (2 + \sqrt{3})$ contains points of the horizon $H$ in its boundary. Since $H$ is generated by null curves, this implies TM-structure.

\subsection*{Geodesics}
As we expect, the geodesic equations of the pseudo Schwarzschild space\-time bear structural similarities to the equations for the usual Schwarzschild spacetime.
Directly computing the Christoffel symbols and assembling the equations results in:
\begin{align*}
 \ddot{\nu} - \frac{m}{r^2} \dot{\nu}^2 - r \dot{\theta}^2 - r \sinh^2(\theta) \dot{\phi}^2  &=0,	  \\
 \ddot{\theta} + \frac{2}{r} \dot{r} \dot{\theta} - \cosh \theta \sinh \theta \dot{\phi}^2 &=0,	  \\
 \ddot{\phi} + \frac{2}{r} \dot{r} \dot{\phi} + 2 \frac{\cosh \theta}{\sinh \theta} \dot{\theta} \dot{\phi} &=0,     \\
 \ddot{r} + \frac{m}{r^3}(r-2m) \dot{\nu}^2 + \frac{2m}{r^2} \dot{\nu} \dot{r} + (r-2m) \dot{\theta}^2 + (r-2m) \sinh^2 (\theta)
 \dot{\phi}^2 &=0. 	 
\end{align*}
Furthermore we have the metric condition resulting from the causal character of the geodesic:
\begin{equation} \label{equ-mcon}
 \left( 1- \frac{2m}{r} \right) \dot{\nu}^2 + 2 \dot{\nu} \dot{r} + r^2 \dot{\theta}^2 + r^2 \sinh^2( \theta) \dot{\phi}^2 =k
\end{equation}
We have $k=-1,0,1$ for a timelike, null or spacelike geodesic, respectively.\\
Due to the homogeneity of the hyperbolic plane, we can assume $\phi=0$ (this is similar to assuming
$\theta = \frac{\pi}{2}$ in the Schwarzschild geometry, where, just the same, it could be $\phi$ which is assumed constant; here we have to use $\phi$, for $\theta>0$ always). Any geodesic is then obtained by applying an appropriate isometry of the hyperbolic plane to a geodesic with $\phi=0$.

The Killing fields
\begin{align*}
K_1 &= \partial_\nu, 	\\ 
K_2 &= -\cos(\phi) \partial_\theta + \sin(\phi) \coth(\theta) \partial_\phi,	
\end{align*}
where $K_1$ is obvious and $K_2$ stems from the symmetries of $H^2$, lead to the first integrals
\begin{align}
 \left( 1- \frac{2m}{r} \right) \dot \nu + \dot r &= \mathcal C, \label{energy} \\
 r^2 \dot \theta &= \mathcal H,					\label{momentum}
\end{align}
and substituting them into the metric condition \eqref{equ-mcon} yields 
\begin{equation} \label{equ-potential-r}
 \mathcal C^2 = \dot r ^2 + V(r)
\end{equation}
with
\begin{equation*} \label{potential}
 V(r) \pdef \frac{2m\mathcal H^2}{r^3} - \frac{\mathcal H^2}{r^2} - \frac{2mk}{r} + k = 
\left( 1- \frac{2m}{r} \right) \left( k- \frac{\mathcal H^2}{r^2} \right).
\end{equation*}
Note that $V(r)$ is exactly the negative of the Schwarzschild potential.
The behaviour of $r$ can now be deduced from the properties of the potential $V$.

\begin{prop} \label{prop-potential}
The potential $V$ satisfies (suppose $V \neq 0$):
\begin{enumerate}
  \item $V(r) = 0 \quad \Leftrightarrow \quad \begin{cases}
                                   	r = 2m 			\quad 	&\text{if} \quad k \neq 1 \text{ or } \mathcal H=0,\\
					r \in \{2m, |\mathcal H|\}		&\text{if} \quad k=1 \text{ and } \mathcal H \neq 0;		
                                  \end{cases}
	$\\[0.2cm]
	For the second case we set $r_0 \pdef \min \{2m,|\mathcal H|\}$ and $r_1 \pdef \max \{2m,|\mathcal H|\}$.
  \item $\lim_{r \to \infty} V(r) = k \quad \text{and} \quad \lim_{r \to 0} V(r) = \begin{cases}
                                                                                    + \infty \quad &\text{if} \quad \mathcal H \neq 0 \text{ or } k=-1,\\
										    - \infty	& \text{if} \quad \mathcal H=0 \text{ and } k=1;
                                                                                   \end{cases}$
  \item $V'(r) = 2r^{-4}(mkr^2 + \mathcal H^2r - 3m\mathcal H^2)$ and if $k\neq 1$, then
	\begin{equation*}
	V'(r) < 0 \qquad  \text{for} \quad  r \leq 2m. 
	\end{equation*}
	If $k=1$ and $\mathcal H \neq 0$, then
	\begin{equation*}
	V'(r) = 0 \quad \Leftrightarrow \quad  r = r_2 \pdef \frac{\mathcal H}{2m} \left( -\mathcal H + \sgn(\mathcal H) \sqrt{\mathcal H^2+12m^2} \right)
	\end{equation*}
	and $r_0 \leq r_2 \leq r_1$.
  \label{prop-potential-derivative}
 \end{enumerate}
\end{prop}

We now suppose that we are given a geodesic $\gamma$ defined on some interval $(a,b)$ around zero approaching the horizon $H$, i.e. $\lim_{s \to b}  r(s) = 2m$ and proceed under the assumptions $V \neq 0$ and $\mathcal C \neq 0$ and that $r(s) \neq 2m$ except at isolated values of $s$. From \eqref{equ-potential-r} and the foregoing proposition, we see that  $\lim_{s \to b} \dot r(s) = \pm \mathcal C$. Therefore, judging from the expression (compare \eqref{energy})
\begin{equation} \label{equ-nu-cross-horizon}
 \dot \nu = \frac{\mathcal C - \dot r}{1-\frac{2m}{r}} = \frac{k-\frac{\mathcal H^2}{r^2}}{\mathcal C+\dot r} \qquad \text{for } r(s) \neq 2m,
\end{equation} 
we distinguish two cases for $\gamma$:\\

\textbf{(I) Geodesics that cross the horizon:}\\
This is the case of $\lim_{s \to b} \dot r(s) = \mathcal C$. Here $\nu$ does not behave singularly 
as the geodesic approaches the horizon. Instead,
\begin{equation*}
 \lim_{s \to b} \dot \nu(s) = \frac{k-\frac{\mathcal H^2}{4m^2}}{2\mathcal C},
\end{equation*}
i.e. $\gamma$ leaves the region $R_1$ and enters $R_2$ or vice versa, as the case may be.\\

\textbf{(II) Incomplete geodesics ending at the horizon:}\\
Here we have $\lim_{s \to b} \dot r(s) = -\mathcal C$. We immediately see from \eqref{equ-nu-cross-horizon} that no matter where the geodesic begins, in $R_1$ or in $R_2$, we have $\lim_{s \to b} \dot \nu (s)= +\infty$. 
Whenever this case occurs in the discussion that follows, it will be clear from the form of the potential that $b<\infty$. Such a geodesic can never be extendible since any extension to an interval $(a,\tilde b)$, $b<\tilde b$ would imply $\lim_{s\to b} \nu(s) < \infty$.
Hence these geodesics are incomplete.\\

Because of this geodesic incompleteness, the embedded $\theta=\textrm{const}$ and $\phi=\textrm{const}$ cylinders resemble more the situation in two-dimensional Misner spacetime \cite{Misner1967} and the pseudo Schwarzschild spacetime behaves globally rather similar to Taub-NUT space (see the analysis in \cite{HawEllis}) than it does to the usual Schwarzschild black hole. Similar to Misner spacetime one could consider the ingoing pseudo Finkelstein transformation of coordinates corresponding to (\ref{newcoord}) (i.e.~$\nu=-\bar{\nu}+r^{\ast}$) resulting in a metric
\[
 \tilde{g}_1 := \left( 1- \frac{2m}{r} \right) d\nu^2 - 2 d\nu dr 
\]
on $Z$. This transformation makes the null geodesics on $Z$ complete that were left incomplete by (\ref{newcoord}) and vice versa. But exactly as in Misner spacetime, and in contrast to Schwarzschild spacetime, one cannot carry out both transformations at the same time and hence it is impossible set up a pseudo Kruskal maximal analytic extension of the pseudo Schwarzschild spacetime which is geodesically complete at the horizon\footnote{However, one can carry out both transformations simultaniously if one is prepared to introduce the notion of a non-Hausdorff spacetime \cite{HawEllis}. But we will not choose this mathematically rather cumbersome path.}. By using Prop.~6.20 in \cite{BeemEhrlich} we can give a rigorous proof of the inextendibility of the pseudo Schwarzschild spacetime: It suffices to analyze incomplete geodesics on some embedded cylinder $Z$ for $\theta=\textrm{const}$ and $\phi=\textrm{const}$. Assume any incomplete geodesic on $Z$ approaching the horizon $H$. As the metric is completely regular on $H$ and $H\cap Z$ is given by $r=2m$, we would have to add a point also obeying $r=2m$ to $Z$ in order to find a local extension of $Z$ that makes the geodesic extendible to $H$. But as $H\cap Z$ is compact, the resulting space would not be a manifold any more. 

Thus we face the problem of imprisoned incomplete geodesics in the pseudo Schwarzschild spacetime. Similar to the situation in Taub-NUT space these incomplete imprisoned geodesics end at the horizon, but do not correspond to any kind of curvature singularity. Following the reasoning of Section 8.5 in \cite{HawEllis} this is due to the special nature of the Riemann curvature tensor of the pseudo Schwarzschild spacetime, which is a vacuum solution of Einstein's equations. If there is only a slight deviation from the vacuum condition, Prop.~8.5.2 in \cite{HawEllis} yields the existence of curvature singularities along the imprisoned incomplete geodesics. As we will show below that all causal geodesics in the pseudo Schwarzschild spacetime are imprisoned incomplete, there follows a potentially quite unpleasant fate of massive observers approaching the time machine on geodesic world lines. Even the disturbance of the vacuum pseudo Schwarzschild spacetime by the observer's own mass may eventually lead to unboundedly large curvature forces along his world line in finite proper time. 

\subsubsection*{Causal Geodesics}
The qualitative behavior of $r$ can be seen from Figure \ref{fig-potential-causal}.
In both cases (A) and (B) indicated there, the geodesic approaches the horizon. 
Depending on the sign of $\mathcal C$ it may cross $H$ or end there. In the interesting case of a causal geodesic 
starting in the region $R_1$ with $\dot r(0) <0$, even if $\gamma$ crosses $H$, 
it will be reflected off at some $r_{min}$ and reach the horizon again.
This time it will surely not cross it but end there, since the sign of $\dot r$ is now reversed.

\begin{figure}[h]  \centering
 \includegraphics[scale=0.5]{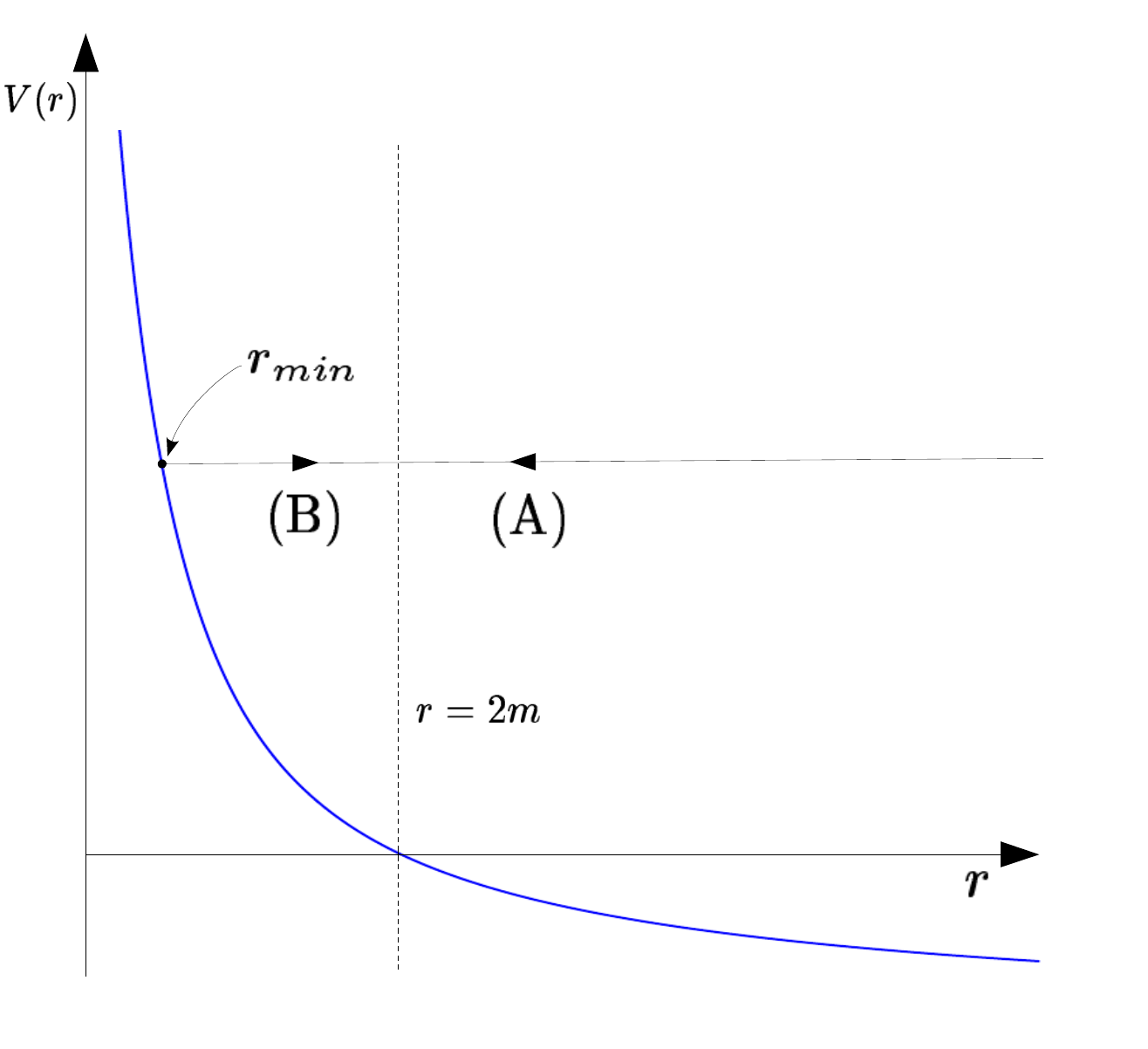}
 \caption{Qualitative form of the potential for causal geodesics.}
\label{fig-potential-causal}
\end{figure}

Therefore, a freely falling observer following a worldline that enters the time machine, will never reach the curvature singularity at $r=0$. Rather the geodesic is turning back and as the compactified $\nu$ direction is timelike in the region $R_2$, this poses the problem of possible self-intersections of extended objects entering the causality violating region on such a track. Similar problems occur in the Misner spacetime (see \cite{Misner1967} and \cite{HawEllis}) and were recently analyzed in \cite{Levanony}. 

So far we have excluded the special cases $\mathcal C=0$ and $V=0$. If we consider $V=0$, we find two sets of null geodesics corresponding to the radial null geodesics in the Schwarzschild spacetime. $V=0$ implies $k=\mathcal H=0$. Hence, $\theta = const$ and $\dot r=\pm \mathcal C$. The plus sign for $\dot r$ leads to the geodesics $s \mapsto (\nu_0, \mathcal Cs+r_0,\theta_0,0)$, i.e. linearly parametrized $r$-coordinate lines are lightlike geodesics, which cross the horizon and enter the time machine. In the case of the minus sign, the geodesic equations imply $\mathcal C =1$ and 
\begin{equation*}
r(s) = s +r_0 \qquad \dot \nu =-\frac{2}{1-\frac{2m}{r}}. 
\end{equation*}
Note that, as we approach the horizon $H$, $\dot \nu$ diverges. Hence the second set of null geodesics do not cross the horizon but end there, so they all are incomplete.\\
If we look at the case $\mathcal C=0$, we have $\dot r = \pm \sqrt{-V(r)}$. This is only possible for $r>2m$ and we deduce from \eqref{energy} that
\begin{equation*}
  \dot \nu = \mp \sqrt{\frac{\frac{\mathcal H^2}{r^2}-k}{1-\frac{2m}{r}}} \qquad \text{for} \quad r > 2m.
\end{equation*}
Thus, we see again that $\gamma$ ends at $H$ since $\nu$ diverges there. It is interesting to note that the $\mathcal{C}=0$ timelike geodesics approaching the horizon correspond to the $\bar{\nu}=\textrm{const}$ geodesics along the $\bar{r}$-coordinate lines in the coordinate representation (\ref{newmetric}). In the barred coordinates the vector field 
\[
 V=-\sqrt{1-\frac{2m}{\bar{r}}}\partial_{\bar{r}}
\] 
is geodesic and its integral curves are the $\bar{r}$-coordinate lines approaching the horizon. Doing the pseudo-Finkelstein transformation (\ref{newcoord}) for $V$ we arrive at the vector field 
\[
 \tilde{V}=\frac{1}{\sqrt{1-\frac{2m}{\bar{r}}}}\partial_{\nu}-\sqrt{1-\frac{2m}{r}}\partial_{r},
\] 
which obviously fulfills the geodesic equations (\ref{energy}), (\ref{momentum}) and (\ref{equ-potential-r}) for $\mathcal{C}=0$ and $\theta=\textrm{const}$. 

In analogy to the usual Schwarzschild black hole, we can consider freely falling observers along $\bar{r}$-coordinate lines as the most natural way of approaching the time machine, without an initial boost in the $\bar{\nu}$, $\theta$ and $\phi$ directions. But the analysis shows that these geodesics are always incomplete. Hence to enter the time machine on a geodesic one has to start with a specific initial boost in the $\bar{\nu}$ direction. In this case one avoids the fate of moving on an incomplete worldline that terminates before the horizon. Prop.~2 below shows that the geodesic nevertheless remains incomplete.

\subsubsection*{Spacelike Geodesics}
In the spacelike case the potential has a different form (see Fig. \ref{fig-potential-spacelike}), leading to different behaviour of $r$. 

\begin{figure}[h] \centering
\includegraphics[scale=0.45]{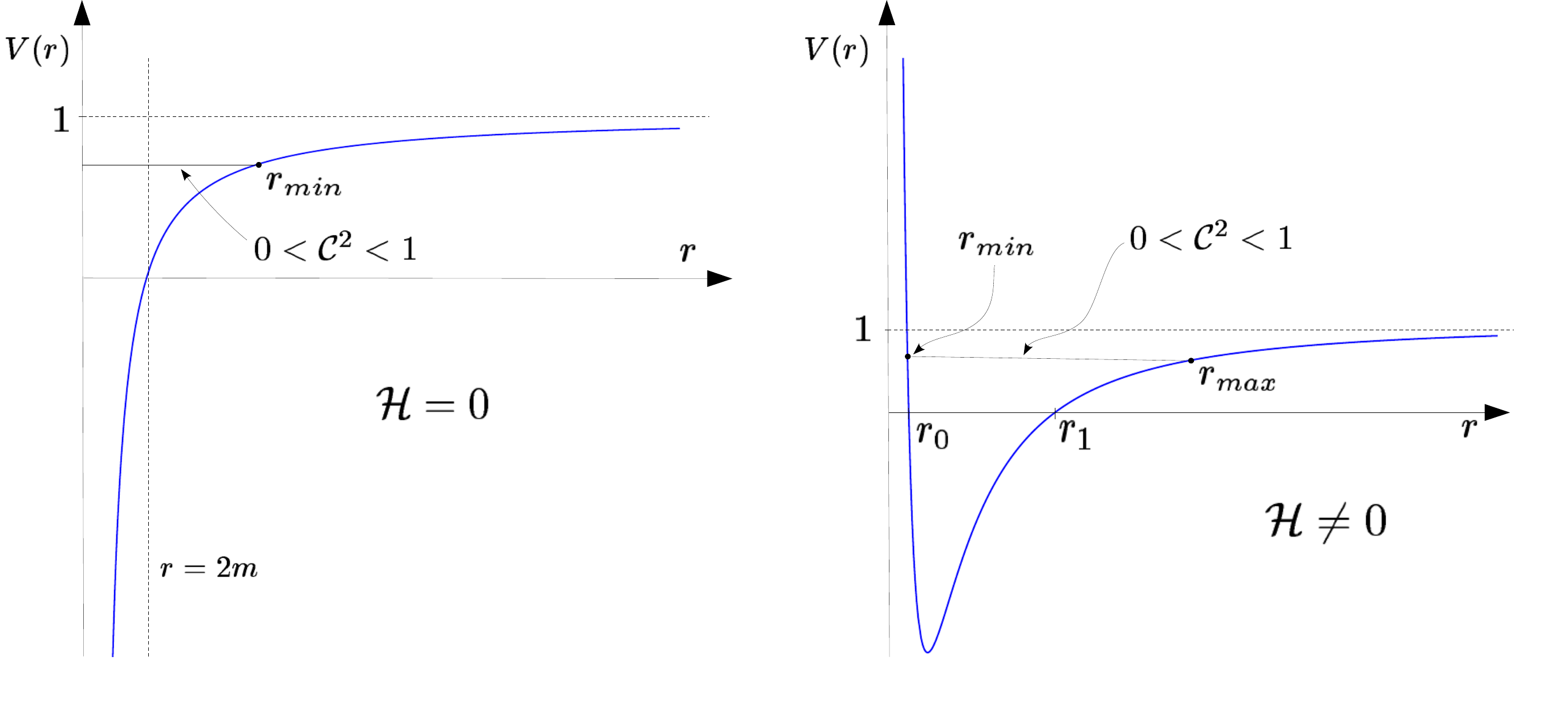}
\caption{Potentials for spacelike geodesics.}
\label{fig-potential-spacelike}
\end{figure} 

As before, let us suppose $\mathcal C \neq 0$. We have to distinguish:
\begin{enumerate}[(a)]
 \item $\mathcal C^2 \geq 1$: We see from the potential that for a geodesic $\gamma$ starting in $R_1$ with $\dot r(0) >0$,  the function $r$ will always be strictly increasing and is unbounded above. If $\dot r(s)<0$ initially, it depends on the sign of $\mathcal C$ whether it crosses the horizon $H$ (I) or ends there (II). This last statement is also true in the case of $\gamma$ starting in $R_2$ and $\dot r(s)>0$ initially. On the other hand, if $\gamma$ starts in $R_2$ and $\dot r(0)<0$ we clearly see that it ends at $r=0$ if $\mathcal H=0$ or reaches a turning point $r=r_{min}<2m$ if $\mathcal H \neq 0$. In the latter case, $\gamma$ is again incomplete because of the sign reversal of $\dot r$ at $r=r_{min}$.

 \item For $\mathcal C^2 <1$ and  $\mathcal H \neq 0$, $r$ would vary periodically between some minimum 	value $r_{min}$ and some maximum value $r_{max}$. As this happens, $\gamma$ has to cross the horizon $H$ since we have $r_{min} < 2m < r_{max}$. Depending on the sign of $\mathcal C$ and $\dot r$, it can happen that we are in case (II) above at the first intent of $\gamma$ to cross $H$. Then $\gamma$ ends there. If not, $\gamma$ crosses $H$, but it is easy to see that then $\gamma$ ends at its second intent to cross $H$ because $\mathcal C$ and $\dot r$ will then always have opposite signs. In both cases $\gamma$ is incomplete.

 \item If $\mathcal C^2 <1$ and $\mathcal H = 0$, a geodesic $\gamma$ with $\dot r(s)>0$ initially starting in $R_1$ reaches a turning point $r_{max}$ and then begins to approach $H$. Now again, as with a geodesic with $\dot r(0)<0$ starting in $R_1$ as well, it depends on the signs of $\mathcal C$ and $\dot r$ whether $\gamma$ crosses $H$ or ends there. If it does cross $H$, it clearly ends at $r=0$. Similar arguments apply to a geodesic $\gamma$ starting in $R_2$. Hence all these geodesics are incomplete.
\end{enumerate}

Let us now consider $\mathcal C=0$ for spacelike geodesics. Equations \eqref{energy} and \eqref{equ-nu-cross-horizon} become
\begin{equation} \label{equ-nu-energy-zero-horizon}
 \dot r ^2 = - V(r) \qquad \text{and} \qquad \dot \nu = -\frac{\dot r}{1-\frac{2m}{r}} = \frac{1-\frac{\mathcal H^2}{r^2}}{\dot r}.
\end{equation}
Consider first $\mathcal H=0$. From the form of $V$ we conclude that $\gamma$ has to start in $R_2$. If $\gamma$ satisfies $\dot r(0)>0$, the second equation in \eqref{equ-nu-energy-zero-horizon} shows that $\gamma$ ends at $H$. If $\dot r(0) <0$, $\gamma$ clearly ends at $r=0$ as in case (c) above.\\
For $\mathcal H \neq 0$ we can suppose $|\mathcal H| \neq 2m$ (for otherwise $V \geq 0$, hence $\dot r = 0$ and $r=2m$). The form of $V$ tells us that $r$ is bounded below by $r_0$ and above by $r_1$. As $\gamma$ approaches $H$, \eqref{equ-nu-energy-zero-horizon} shows that $\dot \nu$ diverges and $\gamma$ ends there. On the other hand, $r = |\mathcal H|$ is a turning point since $\dot \nu$ tends to zero in that case.\\

In our analysis so far we had to exclude the case that $r=2m$. If we assume so, we deduce from \eqref{equ-mcon} that
$\dot \theta =\frac{k}{2m}$ and $k \geq0$. All the geodesic equations are then satisfied, if
\begin{equation*} \label{equ-geodesic-nu-horizon}
 \ddot \nu - \frac{1}{4m} \dot \nu ^2 - \frac{k}{2m}=0.
\end{equation*}
Solving this equation, one comes to the conclusion that each point in $H$ lies on a null geodesic and that all these null geodesics as well as the spacelike geodesics resulting from $k=1$ are incomplete. Hence the horizon is generated by closed incomplete null geodesics.
\\

We can sum up the global properties of the pseudo Schwarzschild spacetime investigated here in the following proposition.

\begin{prop}
The pseudo Schwarzschild spacetime is analytically inextendible and all non-constant geodesics are incomplete, with the non-spacelike ones imprisoned incomplete. 
\end{prop}

\section*{Pseudo Kerr Spacetime}
Just like the Schwarzschild spacetime is a special case of the Kerr spacetime, we can construct a pseudo Kerr spacetime containing an additional parameter $a$ which contains the pseudo Schwarzschild spacetime in the special case $a=0$.\\

To do so, we consider the manifold $M \pdef S^1 \times \R^3$ with metric
\begin{equation*} \label{equ-metric-Pseudo-kerr}
 g \pdef g_{\nu\nu} d\nu^2 -  2 d\nu dr + 2 g_{\nu \phi} d\nu d\phi - 2a \sinh ^2 (\theta) dr d\phi
 + \rho^2 d\theta^2 +g_{\phi \phi} d\phi^2
\end{equation*}
and
\begin{align*}
 g_{\nu \nu} 	& \pdef 1 - \frac{2mr}{\rho^2},	\\
 g_{\nu \phi} 	& \pdef - \frac{2mar}{\rho^2} \sinh^2(\theta),	\\
 g_{\phi \phi} 	& \pdef \left( r^2 + a^2 - \frac{2ma^2r}{\rho^2} \sinh^2(\theta) \right) \sinh^2(\theta),
\end{align*}
where
\begin{equation*}
 \rho^2	 \pdef r^2 + a^2 \cosh^2(\theta).
\end{equation*}
As in the pseudo Schwarzschild case, we think of $\theta$ and $\phi$ as pseudo spherical
polar coordinates, $\nu \in S^1$ is a circular coordinate but here $r \in \R$.
We further suppose $a,m >0$ and $a<m$, corresponding to slow Kerr.\\
If $o$ denotes the origin in $\R^2$, the set $S \pdef S^1 \times \R \times \{o\}$ is the part of $M$ 
that is not covered by the coordinates $(\nu,r,\theta,\phi)$. However, a little inspection shows that
we can extend the metric over $S$ where it restricts to
\begin{equation} \label{equ-metric-Pseudo-kerr-on-S}
 g|_{S} = \left( 1 - \frac{2mr}{r^2+a^2} \right) d\nu^2 -  2 d\nu dr + (r^2 + a^2) g_{H^2}|_{o}.
\end{equation}
The metric of the hyperbolic plane $H^2$ in the last summand is again given by
\begin{equation*}
 g_{H^2} = d\theta^2 + \sinh^2(\theta) d\phi^2.
\end{equation*}
The Kerr function is
\begin{equation*}
 \Delta(r) \pdef r^2 - 2mr + a^2
\end{equation*}
and with our choice for $a$ it has two zeros $ 0 < r_0 < r_1 < 2m $ given by
\begin{equation*}
 r_i \pdef m + (-1)^i \sqrt{m^2 - a^2}.
\end{equation*}
This separates $M$ into the regions $\mathrm{R_I}$, $\mathrm{R_{II}}$, $\mathrm{R_{III}}$ 
for which $r_1<r$, $r_0<r<r_1$ and $r<r_0$, respectively. Furthermore, we have two horizons $H_i$
characterized by $r=r_i$ with their union designated by $H$.\\
It will be useful for many calculations to collect here the following identities.
\begin{prop} \label{prop-metric-identities}
 For the coefficient functions of the pseudo Kerr metric the following identities hold:
\begin{enumerate}[(1)]
 \item $\quad ag_{\phi \phi} + (r^2 + a^2) g_{\nu \phi} = a \Delta \sinh^2(\theta)$
 \item $\quad ag_{\nu \phi} + (r^2 + a^2) g_{\nu \nu} = \Delta$
 \item $\quad g_{\nu \nu} g_{\phi \phi} - g_{\nu \phi}^2 = \Delta \sinh^2(\theta)$ \label{equ-identity-for-determinant}
\end{enumerate}
\end{prop}

We introduce a new coordinate system $(\bar \nu, \bar r, \bar \theta, \bar \phi)$ on the subset $M \setminus (H \cup S)$
which will serve better for some purposes.
It is implicitly defined by the map
\begin{equation*}
(\nu, r, \theta, \phi)  \mapsto (\bar \nu, \bar r, \bar \theta, \bar \phi) = (\nu - T(r), r, \theta, \phi - A(r)),
\label{equ-new-coord-system}
\end{equation*}
where $T$ and $A$ are functions of $r$ satisfying
\begin{equation*}
 \frac{dT}{dr} = \frac{r^2 + a^2}{\Delta} \qquad \text{and} \qquad \frac{dA}{dr} = \frac{a}{\Delta}.
\end{equation*}
%

This transformation has the advantage that it reduces the number of off-diagonal terms in 
the coordinate expression of the metric tensor by one, leaving only $g_{\bar \nu \bar \phi}$ non-vanishing. 
After some calculation, one obtains as non-zero components
\begin{align*} 
 g_{\bar \nu \bar \nu} 	&= 1-\frac{2mr}{\rho^2},		\\
 g_{\bar \nu \bar \phi}	&= -\frac{2mar}{\rho^2} \sinh^2(\theta),	\\
 g_{\bar r \bar r}		&= -\frac{\rho^2}{\Delta},	\\
 g_{\bar \theta \bar \theta} &= \rho^2,			\\
 g_{\bar \phi \bar \phi}	&= \left(r^2 + a^2 - \frac{2ma^2r}{\rho^2}\sinh^2(\theta) \right) \sinh^2(\theta).
\end{align*}
It is in these coordinates that $g$ resembles most the usual Kerr metric in Boyer-Lindquist form.
The difference in the topology of the underlying manifolds of the Kerr and pseudo Kerr metric remains, however.\\
Similar to the situation in the old coordinates above, the expression for the metric tensor shows that
on $S \setminus H$ we can write $g$ as
\begin{equation} \label{equ-metric-Pseudo-kerr-on-S-overbar-coord}
 g|_{S \setminus H} = \left( 1- \frac{2mr}{r^2+a^2} \right) d \bar \nu  \otimes d \bar \nu - \frac{\rho^2}{\Delta} dr \otimes dr + \rho^2 g_{H^2}|_o.
\end{equation}
Note that we have used $d\bar r = dr$ in this expression.\\
Due to the form of the metric components and by analogy to the Kerr and pseudo Schwarzschild metric, 
we time orient the region $\mathrm{R_I}$ by requiring $-\partial_{\bar r}$ to be future pointing on $\mathrm {R_I}$.
This leads to the fact that a causal curve $\alpha \colon I \to M$ is future pointing in $\mathrm{R_I}$ if and only if $\dot{\bar r} = \dot r <0$ for $\alpha$. To see this, we only need to look at 
\begin{equation*}
 \langle \alpha', -\partial_{\bar r} \rangle = \frac{\rho^2}{\Delta} \dot{\bar r} < 0 \qquad \Leftrightarrow \qquad \alpha \text{ future  pointing in } \mathrm{R_I}
\end{equation*}
and the fact that $\Delta > 0$ in $\mathrm{R_I}$.\\
The usual Kerr metric has the interesting feature of admitting a Killing tensor which is used in analysing geodesics. This property is still present in the pseudo Kerr metric. Using the overbar coordinates, we define in component notation
\begin{equation*}
 l^{\mu} \pdef \left(\frac{r^2+a^2}{\Delta}, 1,0,\frac{a}{\Delta} \right) \qquad \text{and} \qquad
 n^{\mu} \pdef \frac{1}{2\rho^2}\left(r^2+a^2,-\Delta,0,a\right).
\end{equation*}
It can then be verified that the symmetric $(2,0)$-tensor
\begin{equation*}
 \kappa_{\mu \nu} \pdef 2\rho^2 l_{(\mu} n_{\nu)} - r^2 g_{\mu \nu}
\end{equation*}
is a Killing tensor for the pseudo Kerr metric. The difference to the corresponding Killing
tensor of the Kerr metric resides in the minus sign before the second summand in the definition.\\
Note also, that in the pseudo Kerr spacetime there is no longer any analogue to the ring singularity of the ususal Kerr spacetime.

\subsection*{CTCs in pseudo Kerr spacetime}
\begin{prop}
 There exists a closed timelike curve through the point $p \in M$ if and only if $ p \in \mathrm{R_{II}}$.
\end{prop}
\begin{proof}
 The causal character of the hypersurfaces $r = const$ is determined by
\begin{equation*}
 g^{rr} = - \frac{\Delta(r)}{\rho^2}.
\end{equation*}
Hence they are spacelike if $\Delta(r) > 0$ making it impossible for any CTC to reach any of the regions $R_{\mathrm I}$ or $R_{\mathrm{II}}$. For otherwise such a CTC would have to be tangent to some hypersuface $r=const$.\\
Next, we consider points $p \in \mathrm{R_{II}} \setminus S$ which allows us to use the coordinates $(\nu , r, \theta, \phi)$. If $r$ and $\theta$ are such, that
\begin{equation*}
 g_{\nu \nu} < 0 \qquad \Leftrightarrow \qquad \Delta(r) + a^2 \sinh^2(\theta) <0,
\end{equation*}
which is the case for $\theta$ small, then the closed curves $s \mapsto (s,r,\theta,\phi)$ are CTCs for any $\phi$.
Similarly, if $r$ and $\theta$ satisfy the condition
\begin{equation*}
 g_{\phi \phi} <0 \qquad \Leftrightarrow \qquad \Delta(r) + 2mr \left( 1- \frac{a^2 \sinh^2(\theta)}{r^2 + a^2 \cosh^2(\theta)} \right) < 0,
\end{equation*}
which holds for $\theta$ large enough, the closed curves $s \mapsto (\nu, r, \theta, s)$ are CTCs for any $\nu$.
Consider now the case of $g_{\nu \nu},g_{\phi \phi} \geq 0$. The curve 
\begin{equation} \label{equ-CTC-positive-gnunu-gphiphi}
 I \to M, \qquad s \mapsto (us, r, \theta, s),
\end{equation}
for an appropriate interval $I$ and a constant $u \neq 0$, is timelike iff
\begin{equation*}
 g_{\phi \phi} + 2g_{\nu \phi} u + g_{\nu \nu} u^2 < 0.
\end{equation*}
This equation has zeros
\begin{equation*}
 c_\pm \pdef g_{\nu \nu}^{-1} ( -g_{\nu \phi} \pm \sqrt{-\Delta(r)}\sinh \theta), \qquad \qquad g_{\nu \nu} \neq 0.
\end{equation*}
Hence any rational $u$ satisfying
\begin{equation*}
\begin{cases}
 u \in (c_- , c_+ ) \qquad &\text{if} \quad g_{\nu \nu} \neq 0, \\
 u > -\frac{g_{\phi \phi}}{2g_{\nu \phi}} & \text{if} \quad g_{\nu \nu} = 0
\end{cases}
\end{equation*}
converts \eqref{equ-CTC-positive-gnunu-gphiphi} into a CTC.\\ 
Summarizing, we have shown, that given any values of $r$ and $\theta$, there exists a CTC 
through the point $(0, r, \theta, 0)$. But the maps $(\nu,r,\theta,\phi) \mapsto (\nu + \nu_0, r, \theta, \phi)$ 
and $(\nu,r,\theta,\phi) \mapsto (\nu, r, \theta, \phi + \phi_0)$ are isometries of $M$. 
Hence any $p \in \mathrm{R_{II}} \setminus S$ lies on some CTC.	\\ 
Suppose now, that $p \in \mathrm{R_{II}} \cap S$. On $S$, the metric tensor has the form \eqref{equ-metric-Pseudo-kerr-on-S} and since $\Delta(p) < 0$, we have $g_{\nu \nu}(p) < 0$. Therefore the curve $\alpha \colon [-\pi,\pi] \to M$ determined by
\begin{equation*}
 \nu(\alpha(s)) = s \qquad \text{and} \qquad \Pr(\alpha(s)) = \Pr(p) \qquad \text{for all } s,
\end{equation*}
where $\Pr \colon M \to \R^3$ is the projection of $M$ onto its topological factor $\R^3$, is a CTC through $p$.
\end{proof}

\subsection*{Time Machine Structure}
We begin by showing that $H_1$ is generated by null curves.
\begin{prop}\label{prop-H1-generated-by-null-curves}
 Any $p \in H_1$ lies on a null curve entirely contained in $H_1$.
\end{prop}
\begin{proof}
Consider the null curves $\alpha \colon \R \to M$ starting at $p$ that are given by
\begin{equation} \label{equ-closed-null-curve-H1}
\Pr(\alpha(s)) = (r_1,o), \quad \nu(\alpha(s))=s + \nu(p) \quad \text{for all} \quad s \in \R, \; p \in H_1 \cap S
\end{equation}
or
\begin{equation*}\label{equ-closed-null-curves-H1}
 s \mapsto (\eta s + \nu(p), r_1, \theta(p), s + \phi(p)), \qquad \text{for} \qquad p \in H_1 \setminus S,
\end{equation*}
where
\begin{equation*}
 \eta \pdef \frac{2mr_1}{a} = \frac{2m}{a} \left( m + \sqrt{m^2 - a^2} \right)
\end{equation*}
and $\Pr$, as before, is the projection of $M$ onto $\R^3$. None of these curves leaves $H_1$ and the first one is a null curve since $g_{\nu \nu} =0$ on $H_1 \cap S$. The second set of curves consists of null curves because the only solution of 
\begin{equation*}
 g_{\phi \phi} + 2g_{\nu \phi}u + g_{\nu \nu} u^2 = 0, \qquad u \in \R ,
\end{equation*}
is given by $\eta = - \frac{g_{\nu \phi}}{g_{\nu \nu}}$ (we have $g_{\nu \nu} \neq 0$ on $H_1 \setminus S$).
\end{proof}
From now on, all calculations and arguments will be done using the coordinate system 
$(\bar \nu, \bar r, \bar \theta, \bar \phi)$ on $\mathrm{R_I}$. Since there will be no further exceptions, 
we shall drop the overbars in denoting these coordinate functions to render the formulas more accessible. 
Let us fix some $r_s > r_1$ that determines the spacelike hypersurface 
$\mathscr S \pdef \{ p \in M \colon r(p) = r_s\}$ and consider the subset
\begin{equation*}
 \mathscr C \pdef \{ p \in M \colon r(p) = r_s, \; \theta(p) \leq \theta_0 \} \subset \mathscr S \subset \mathrm{R_I}.
\end{equation*}
The constant $\theta_0 >0$ will be subject to constraints, such that the future domain of dependence 
$D^+(\mathscr C)$ contains points on the horizon $H_1$ in its closure (hence in its boundary). 

For a description of $D^+(\mathscr C)$ we need the functions
\begin{equation*}
 b_c \colon [r_1,r_s] \to \R, \qquad r \mapsto \ln(r-m + \sqrt{\Delta(r)}) + c
\end{equation*}
with $c \in \R$ being a constant. They all satisfy
\begin{equation} \label{equ-diffequ-boundary-functions}
 b_c'(r)=\frac{1}{\sqrt{\Delta(r)}}.
\end{equation}
The function $b_{\hat c}$, where $\hat c$ is such, that $b_{\hat c}(r_s) = \theta_0$, 
is called the boundary function. We use it to define the boundary curve
\begin{equation*}
 \beta \colon [r_1,r_s] \to M, \qquad s \mapsto (0,s,b_{\hat c}(s),0).
\end{equation*}
Using this terminology we have:
\begin{prop} \label{prop-future-domain-dependence-Pseudo-Kerr}
 The future domain of dependence of $\mathscr C$ is
\begin{equation} \label{equ-future-domain-of-dep-of-C}
 D^+(\mathscr C) = \{ p \in M \colon r(p) \in (r_1, r_s], \; \theta(p) \leq b_{\hat c}(r(p)) \}.
\end{equation}
\end{prop}

\begin{proof}
 Let $\mathscr D$ denote the set on the right hand side. Suppose $p \in \mathscr D$ and let $\alpha \colon [0,B) \to M$ be a past pointing inextendible causal curve starting at $p$ (here $B \leq \infty$). As we have seen, the fact that $\alpha$ is past pointing implies $\dot r >0$ for $\alpha$; thus 
\begin{equation*}
 \alpha[0,B) \subset \mathrm{R_I}
\end{equation*}
and therefore the metric components $g_{\nu \nu}$ and $g_{\phi \phi}$ in the expression
\begin{equation*}
 \langle \alpha', \alpha' \rangle = -\frac{\rho^2}{\Delta} \dot r^2 + \rho^2 \dot \theta^2 + \left[ g_{\nu \nu} \dot \nu^2 + g_{\phi \phi} \dot \phi^2 +2g_{\nu \phi} \dot \nu \dot \phi \right]
\end{equation*}
are positive along $\alpha$. Because of Prop. \ref{prop-metric-identities}, \eqref{equ-identity-for-determinant} and $\Delta > 0$ in $\mathrm {R_I}$, the term in square brackets is non-negative and we conclude
\begin{equation*}
 \alpha \text{ causal} \qquad \Rightarrow \qquad -\frac{\rho^2}{\Delta} \dot r^2 + \rho^2 \dot \theta^2 \leq 0 \qquad \text{for all } s\in [0,B).
\end{equation*}
This last inequality can be written as
\begin{equation*} \label{equ-inequ-theta-r-Delta}
 \frac{\dot \theta^2}{\dot r^2} \leq \frac{1}{\Delta}.
\end{equation*}
Although we used the coordinate system $(\nu,r,\theta,\phi)$ to deduce this result, an argument using continuity shows that it holds on $S$, too.
Now, the function $s \mapsto r(s)$ is strictly increasing, hence invertible. This allows us to write $\theta$ as a function of $r$ and then the chain rule implies
\begin{equation} \label{equ-inequality-for-theta}
 \theta'(r)^2 = \frac{\dot \theta^2}{\dot r^2} \leq \frac{1}{\Delta} \qquad \Rightarrow \qquad |\theta'(r)| \leq \frac{1}{\sqrt{\Delta}}.
\end{equation}
In view of \eqref{equ-diffequ-boundary-functions}, this secures that $\alpha$ meets $\mathscr C$ if it meets $\mathscr S$ at all.\\
We will show now, that $\alpha$ does indeed meet $\mathscr S$. In fact, we show that $r$ is unbounded above on $[0,B)$. Suppose the contrary, i.e. $r(s) \leq \tilde r$ for all $s \in [0,B)$ and some $\tilde r >0$. Then $\tilde r_0 \pdef \lim_{s \to B} r(s) < \infty$ exists. Furthermore, $\lim_{r \to \tilde r} \theta (r)$ exists, for if it did not exist, $\theta'$ would be unbounded which contradicts \eqref{equ-inequality-for-theta}. But this implies that $\lim_{s \to B} \theta(s)$ exists. Using the inverse function of $s \mapsto r(s)$, we write
\begin{equation} \label{equ-bound-on-derivatives}
 \rho^2 \theta'(r)^2 + g_{\nu \nu} \nu'(r)^2 + g_{\phi \phi} \phi'(r)^2 + 2 g_{\nu \phi} \nu'(r) \phi'(r) \leq \frac{\rho^2}{\Delta}
\end{equation}
for points $\alpha(s) \in \mathrm{R_I} \setminus S$. Assume for the moment $\lim_{s \to B} \theta (s) \neq 0$. Then we are allowed to use \eqref{equ-bound-on-derivatives} for all $s$ sufficiently close to $B$. Because the r.h.s. of \eqref{equ-bound-on-derivatives} is bounded, the limits of $\nu$ and $\phi$ as $s$ tends to $B$ exist as well, showing that $\alpha$ would be extendible.\\
If $\lim_{s \to B} \theta (s) =0$, we use \eqref{equ-metric-Pseudo-kerr-on-S-overbar-coord} to deduce
\begin{equation*}
 \left(1 - \frac{2mr}{r^2 + a^2} \right) \nu'(r)^2 \leq \frac{\rho^2}{\Delta}
\end{equation*}
for points $\alpha(s)$ on $S$. Together with \eqref{equ-bound-on-derivatives} for $\alpha(s) \notin S$, this inequality shows that $\nu'$ is bounded and therefore the limit $\lim_{s\to B} \nu(s)$ exists. The existence of the limits for $r$ and $\theta$ as $s$ approaches $B$ shows that $\lim_{s \to B} \Pr(\alpha(s))$ exists. All in all, $\alpha$ would again be extendible.\\
It remains to show that $D^+(\mathscr C)$ does not contain any points of $M \setminus \mathscr D$. For $p \in M \setminus \mathscr D$ with $r(p) > r_s$ the outward going (i.e. $r$ increasing) $r$-coordinate line from $r(p)$ onwards is a past pointing causal curve that never meets $\mathscr C$. Hence in that case $p \notin D^+(\mathscr C)$. Nor can any $p$ with $r(p) \leq r_1$ lie in $D^+(\mathscr C)$ because by Prop. \ref{prop-H1-generated-by-null-curves} each point $p \in H_1$ lies on a null curve entirely contained in $H_1$ and a past inextendible curve $\alpha$ with $r(\alpha(0)) < r_1$ would have to cross $H_1$ before it reaches $\mathscr C$. A point $p$ violating the second condition in \eqref{equ-future-domain-of-dep-of-C}, i.e. $\theta(p) > b_{\hat c}(r(p))$ cannot lie in  $D^+(\mathscr C)$ since it is the starting point of the past pointing null curve
\begin{equation*}
 s \mapsto (0, r(p) + s, b_{\tilde c}(r(p) +s), 0),
\end{equation*}
where $\tilde c$ is chosen so that $b_{\tilde c}(r(p)) = \theta(p)$. Both functions $b_{\hat c}$ and $b_{\tilde c}$ satisfy \eqref{equ-diffequ-boundary-functions}, but $b_{\tilde c}(r(p)) = \theta(p) > b_{\hat c}(r(p))$. Hence at $r_s$ we have again $b_{\tilde c}(r_s) > b_{\hat c}(r_s)$ and therefore this curve does not meet $\mathscr C$ (but it does meet $\mathscr S$).
\end{proof}

\begin{figure}[h] \centering
 \includegraphics[scale=0.8]{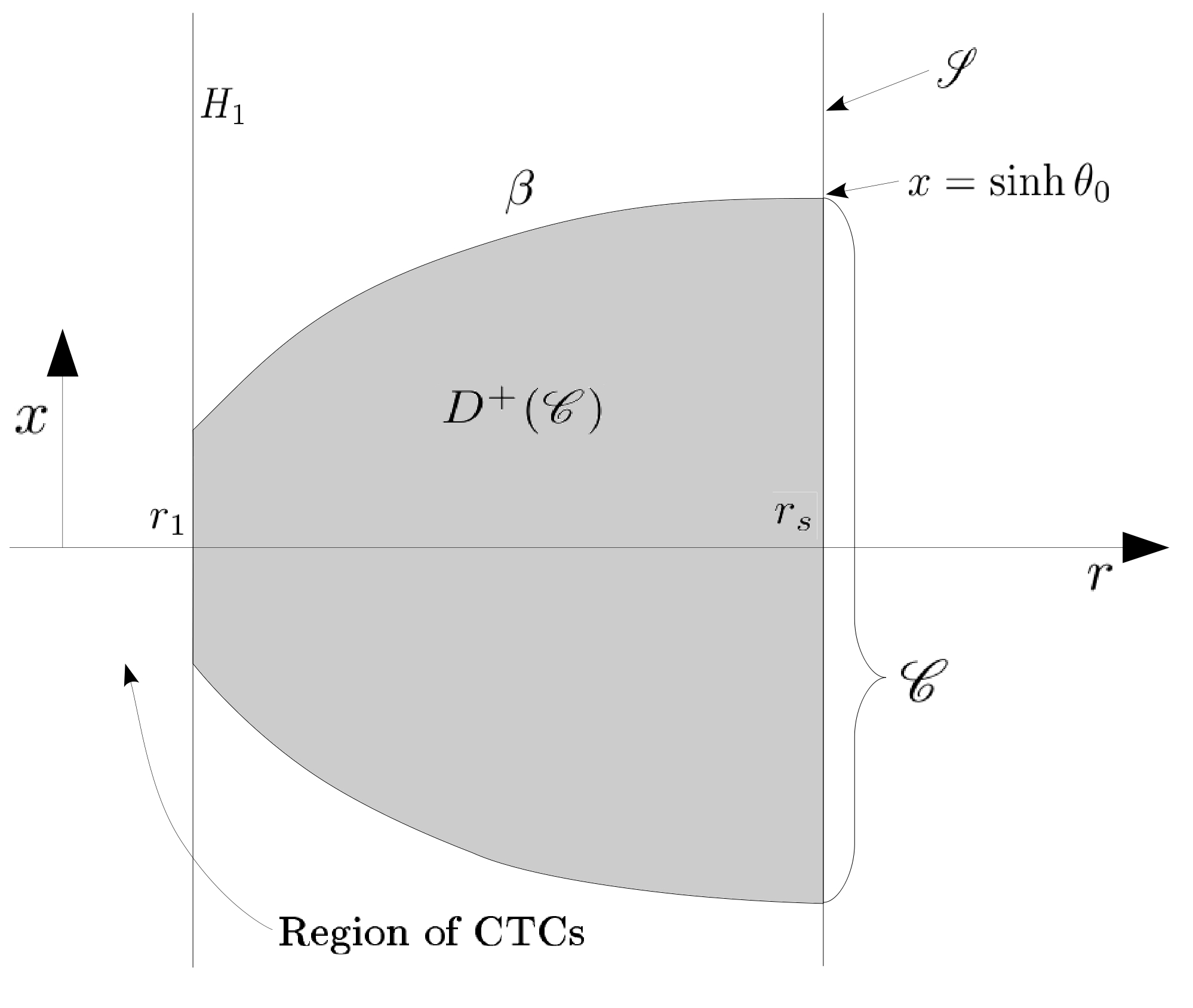}
 \caption{The future domain of dependence of $\mathscr C$ contains points of the horizon $H_1$ in its boundary.}
\end{figure}

In order to make sure that $\overline{D^+(\mathscr C)} \cap H_1 \neq \emptyset$ 
we require $b_{\hat c}(r_1) > 0$. This condition is equivalent to 
\begin{equation*}
 \theta_0 > \ln \frac{r_s - m + \sqrt{\Delta(r_s)}}{r_1-m} > 0.
\end{equation*}
Together with the fact that the null curve \eqref{equ-closed-null-curve-H1} is closed and has non-empty intersection with $\overline{D^+(\mathscr C)}$ we have proved:
\begin{prop}
 The pseudo Kerr spacetime $M$ has TM-structure.
\end{prop}

\bibliographystyle{english_custom}

\bibliographystyle{english_custom}
\end{document}